\documentclass[11pt]{article}

\usepackage[letterpaper,left=1in,right=1in,top=1in,bottom=1in]{geometry}
\usepackage{picins}
\usepackage{floatflt}
\usepackage{newalg}
\usepackage{xspace,pst-node}

\newcommand{\A}{A}
\newcommand{\J}{J}
\newcommand{\first}{f}
\newcommand{\second}{s}

\newcommand{\minlabel}{\ensuremath{\lambda_{\mathrm{min}}}}
\newcommand{\equivlabel}{\ensuremath{\lambda_{\mathrm{equiv}}}}
\newcommand{\lab}{\ensuremath{\lambda}}
\newcommand{\wrt}{{w.r.t.}\xspace}
\newcommand{\etal} {{\it et al. }} 

\usepackage{amsfonts,amsmath,amsthm}
\newtheorem{theorem}{Theorem}
\newtheorem{lemma}{Lemma}
\newtheorem{definition}{Definition}
\newtheorem{invariant}{Invariant}
\newtheorem{observation}{Observation}

\begin{document}

\title{Weighted Popular Matchings \thanks{Research supported by NSF Awards CCR 0113192 and CCF 0430650} }

\author{Juli\'{a}n Mestre \\[1em]
Department of Computer Science \\
 University of Maryland, College Park, MD 20742 \\
{\tt jmestre@cs.umd.edu }
}

\date{}

\maketitle

\begin{abstract}

We study the problem of assigning jobs to applicants. Each applicant has a weight and provides a \emph{preference list}, which may contain ties, ranking a subset of the jobs. An applicant $x$ may prefer one matching over the other (or be indifferent between them, in case of a tie) based on the jobs $x$ gets in the two matchings and $x$'s personal preference. A matching $M$ is \emph{popular} if there is no other matching $M'$ such that the weight of the applicants who prefer $M'$ over $M$ exceeds the weight of those who prefer $M$ over $M'$.

We present algorithms to find a popular matching, or if none exists, to establish so. For instances with strict preference lists we give an $O(n+m)$ time algorithm. For preference lists with ties we give a more involved algorithm that solves the problem in $O( \min(k \sqrt{n}, n) m)$ time, where $k$ is the number of distinct weights the applicants are given.

\end{abstract}

\section{Introduction}

Consider the problem of assigning jobs to applicants where every applicant provides a \emph{preference list}, which may contain ties, ranking a subset of the jobs. More formally, an instance consists of a bipartite graph $H=(\A, \J, E)$ with $n$ vertices and $m$ edges between a set of applicants $\A$ and a set of jobs $\J$. The edge $(x,p)$ belongs to $E$ if job $p$ is on $x$'s preference list. Moreover, every edge $(x,p)$ is assigned a rank $r_x(p) \in Z^+$ encoding the fact that $p$ is $x$'s $r_x(p)$th choice. An applicant $x$ is said to \emph{prefer} job $p$ over $q$ if the edge $(x,p)$ is ranked higher than $(x,q)$, i.e., $r_x(p) < r_x(q)$. If $r_x(p) = r_x(q)$ we have a tie, and we say $x$ is \emph{indifferent} between $p$ and $q$. Likewise, we say $x$ prefers one matching over the other or is indifferent between them based on the jobs $x$ is assigned by the two matchings. Our ultimate goal is to produce a ``good'' matching in $H$.

Despite its simplicity, this framework captures many real-world problems such as the assignment of government-subsidized houses to families \cite{Y96}, the assignment of graduates to training position \cite{HZ79}, and rental markets such as NetFlix \cite{ACKM06} where DVDs must be assigned to subscribers. The issue of what constitutes a fair or good assignment has been studied in the Economics literature \cite{AS98,Y96,Z90}. The least restrictive definition of optimality is that of a \emph{Pareto optimal} matching \cite{ACMM04,AS98}. A matching $M$ is Pareto optimal if there is no matching $M'$ such that at least one person prefers $M'$ over $M$ and nobody prefers $M$ over $M'$. In this paper we study a stronger definition of optimality, that of \emph{popular matchings}. We say $M_1$ is \emph{more popular than} $M_2$ if the applicants who prefer $M_1$ over $M_2$ outnumber those who prefer $M_2$ over $M_1$. A matching $M$ is \emph{popular} if there is no matching more popular than $M$.

Popular matchings were first considered by Gardenfors \cite{G75} who showed that not every instance allows a popular matching. Abraham \etal \cite{AIKM05} gave the first polynomial time algorithms to determine if a popular matching exists and if so, to produce one: An $O(n+m)$ time algorithm for the special case of strict preference lists, and an $O( \sqrt{n} m)$ time algorithm for the general case where ties are allowed. They noted that maximum cardinality matching can be reduced to finding a popular matching in an instance with ties (by letting every edge be of rank 1) thus a linear time algorithm for the general case seems unlikely.

Notice that this definition of popular matching does not make any distinction between the individuals---the opinion of every applicant is valued equally. But what if we had some preferred set of applicants that we would like to give priority over the rest? This option becomes particularly interesting when jobs are scarce or there is a lot of contention for a few good jobs.

To answer this question we propose a new definition for the more popular than relation under which every applicant $x$ is given a positive weight $w(x)$. The \emph{satisfaction} of $M_1$ with respect to $M_2$ is defined as the weight of the applicants that prefer $M_1$ over $M_2$ minus the weight of those who prefer $M_2$ over $M_1$. Then $M_1$ is more popular than $M_2$ if the satisfaction of $M_1$ \wrt $M_2$ is positive. We believe that this is an interesting generalization of popular matchings that addresses the natural need to assign priorities (weights) to the applicants while retaining the one-sided preferences of the original setup.

In this paper we develop algorithms to determine if a given instance allows a weighted popular matching, and if so, to produce one. For the case of strict preference lists we give an $O(n+m)$ time algorithm. When ties are allowed the problem becomes more involved; a second algorithm solves the general case in $O( \min ( k \sqrt{n}, n) m)$ time, where $k$ is the number of distinct weights that the applicants are given.

Our approach is based on deriving a more algorithmic-friendly characterization of popular matchings. Following the line of attack of Abraham \etal \cite{AIKM05} for unweighted instances, we define the notion of \emph{well-formed matchings} and show that every popular matching is well-formed. For unweighted instances one can show \cite{AIKM05} that every well-formed matching is popular. For weighted instances, however, there may be well-formed matchings that are not popular. Our main contribution is to show that these non-popular well-formed matchings can be weeded out by pruning certain bad edges that cannot be part of any popular matching. In other words, we show that the instance can be pruned so that a matching is popular if and only if it is well-formed and is contained in the pruned instance.

\subsection{Related work}

Following the publication of the work of Abraham \etal \cite{AIKM05}, the topic of unweighted popular matchings has been further explored in many interesting directions. Suppose we want to go from an arbitrary matching to some popular matching by a sequence of matchings each more popular than the previous; Abraham and Kavitha~\cite{AK06} showed that there is always a sequence of length at most two and gave a linear time algorithm to find it. One of the main drawbacks of popular matchings is that they may not always exist; Mahdian~\cite{M06} nicely addressed this issue by showing that the probability that a random instance admits a popular matching depends on the ratio $\alpha = \frac{|J|}{|A|}$, and exhibits a phase transition around $\alpha^* \approx 1.42$. Motivated by a house allocation application, Manlove and Sng \cite{MS06} gave fast algorithms for popular assignments with capacities on the jobs.

A closely related, but not equivalent, problem is that of computing a \emph{rank-maximal matching}. Here we want to maximize the number of rank 1 edges, and subject to this, maximize the number of rank 2 edges, and so on. Irving \etal \cite{IKMMP06} showed how to solve this problem in $O( \min(C \sqrt{n}), n) m)$ time where $C$ is the rank of the lowest ranked edge in a rank-maximal matching, while Kavitha and Shah \cite{KS06} gave a faster algorithm for dense instances that runs in $O(Cn^\omega)$, where $\omega < 2.376$ is the exponent for matrix multiplication.

\section{Strict preference lists}

This section focuses on instances where the preference lists provided by the applicants are strict but need not be complete. In order to ease the analysis we first modify the given instance: For each applicant $x$ create a last resort job $l(x)$ and place it at the end of $x$'s preference list. This modification does not affect whether the instance has a popular matching or not, but it does force every popular matching to be applicant complete.

Before proceeding we need a few definitions. Let us partition $\A$ into categories $C_1, C_2, \ldots, C_k$, such that the weight of applicants in category $C_i$ is $w_i$ and $w_1 > w_2 > \ldots > w_k > 0$. Given a matching $M$ we say a node $u$ is \emph{matched} in $M$ if there exits $v$ such that $(u,v) \in M$, otherwise $u$ is \emph{free}. We denote the mate of a matched node $u$ by $M(u)$.

The plan is to develop an alternative characterization for popular matchings that will allow us to efficiently test if a given instance admits a popular matching, and if so to produce one. 

\begin{definition} For every applicant $x \in C_1$ let $\first(x)$ be the first job on $x$'s preference list, we say this is an $\first_1$-job. For $x \in C_{i>1}$ define $\first(x)$ as the first non-$\first_{j<i}$-job on $x$'s list, this is an $\first_i$-job.
\end{definition} 

\begin{definition} For every $x \in C_i$ let $\second(x)$ be the first non-$\first_{j \leq i}$-job on $x$'s list.
\end{definition}

Notice that $\second(x)$ is ill defined when $\first(x) = l(x)$. This is not a problem since, as we will see shortly, the job $\second(x)$ is assigned to $x$ only when there is contention for $\first(x)$, which by definition never happens when $\first(x) =l(x)$. The following properties about first and second jobs are easy to check:

\begin{observation} The set of $\first_i$-jobs is disjoint from the set of $\first_j$-jobs for $i \neq j$.
\end{observation}
\begin{observation}  The set of $\first_i$-jobs is disjoint from the set of  $\second_j$-jobs for $i \leq j$, but may not be for $i > j$.
\end{observation}

\begin{figure}[t]
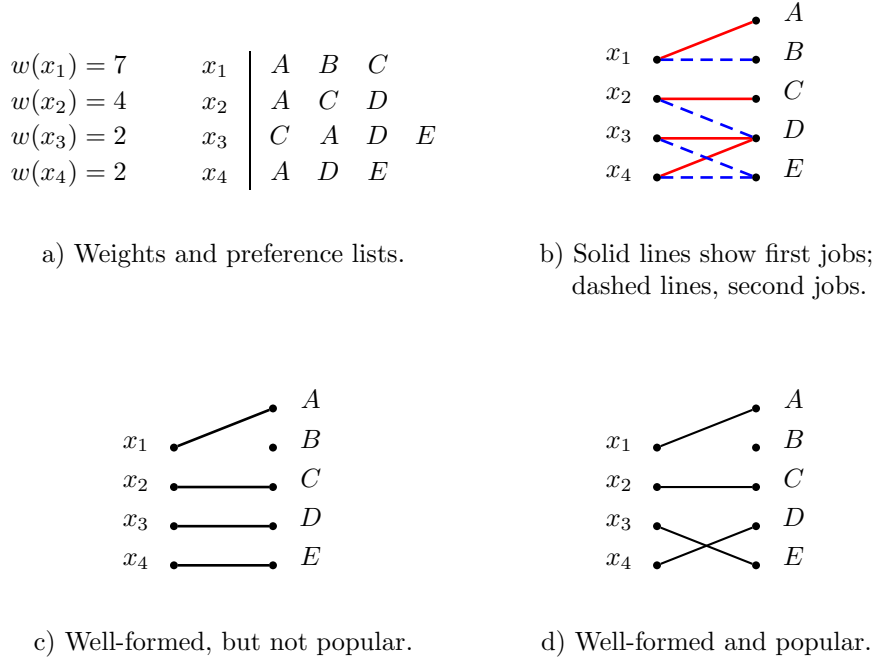

\small
\begin{center}
\begin{tabular}{*{2}{p{6cm}}}

\[\begin{array}{lll|rrrr}
w(x_1) = 7 & \hspace{0.3cm} & x_1 \hspace{0.1cm}&\hspace{0.1cm} A & B & C  \\[0.05cm]
w(x_2) = 4 & & x_2 & A & C & D  \\[0.05cm]
w(x_3) = 2 & & x_3 & C & A & D & E  \\[0.05cm]
w(x_4) = 2 & & x_4 & A & D & E
\end{array}\]

& 
\vspace{-1em}
\[\begin{array}{ll}
 & \cnode*{1.5pt}{a} \rput[B](0.5,0){A} \\[0.1cm]
\cnode*{1.5pt}{x1} \rput[B](-0.5,0){x_1}  \hspace{2.75em} & \cnode*{1.5pt}{b} \rput[B](0.5,0){B} \\[0.1cm]
\cnode*{1.5pt}{x2} \rput[B](-0.5,0){x_2} & \cnode*{1.5pt}{c} \rput[B](0.5,0){C} \\[0.1cm]
\cnode*{1.5pt}{x3} \rput[B](-0.5,0){x_3} & \cnode*{1.5pt}{d} \rput[B](0.5,0){D} \\[0.1cm]
\cnode*{1.5pt}{x4} \rput[B](-0.5,0){x_4} & \cnode*{1.5pt}{e} \rput[B](0.5,0){E}
\end{array}\]
\psset{linewidth=1pt, linecolor=red}
\ncline{-}{x1}{a}
\ncline{-}{x2}{c}
\ncline{-}{x3}{d}
\ncline{-}{x4}{d}
\psset{linewidth=1pt,linecolor=blue,linestyle=dashed}
\ncline{-}{x1}{b}
\ncline{-}{x2}{d}
\ncline{-}{x3}{e}
\ncline{-}{x4}{e}
\\[-1em]
\begin{center}
a) Weights and preference lists.
\end{center} &
\begin{center}
b) Solid lines show first jobs; \\ \hspace{1.5ex} dashed lines, second jobs. 
\end{center}
\\[1em]

\[\begin{array}{ll}
 & \cnode*{1.5pt}{a} \rput[B](0.5,0){A} \\[0.1cm]
\cnode*{1.5pt}{x1} \rput[B](-0.5,0){x_1}  \hspace{2.75em} & \cnode*{1.5pt}{b} \rput[B](0.5,0){B} \\[0.1cm]
\cnode*{1.5pt}{x2} \rput[B](-0.5,0){x_2}   & \cnode*{1.5pt}{c} \rput[B](0.5,0){C} \\[0.1cm]
\cnode*{1.5pt}{x3} \rput[B](-0.5,0){x_3}   & \cnode*{1.5pt}{d} \rput[B](0.5,0){D} \\[0.1cm]
\cnode*{1.5pt}{x4} \rput[B](-0.5,0){x_4}   & \cnode*{1.5pt}{e} \rput[B](0.5,0){E}
\end{array}\]
\psset{linewidth=1pt, linecolor=black}
\ncline{-}{x1}{a}
\ncline{-}{x2}{c}
\ncline{-}{x3}{d}
\ncline{-}{x4}{e}
&
\[\begin{array}{ll}
 & \cnode*{1.5pt}{a} \rput[B](0.5,0){A} \\[0.1cm]
\cnode*{1.5pt}{x1} \rput[B](-0.5,0){x_1}  \hspace{2.75em}  & \cnode*{1.5pt}{b} \rput[B](0.5,0){B} \\[0.1cm]
\cnode*{1.5pt}{x2} \rput[B](-0.5,0){x_2}   & \cnode*{1.5pt}{c} \rput[B](0.5,0){C} \\[0.1cm]
\cnode*{1.5pt}{x3} \rput[B](-0.5,0){x_3}   & \cnode*{1.5pt}{d} \rput[B](0.5,0){D} \\[0.1cm]
\cnode*{1.5pt}{x4} \rput[B](-0.5,0){x_4}   & \cnode*{1.5pt}{e} \rput[B](0.5,0){E}
\end{array}\]
\ncline{-}{x1}{a}
\ncline{-}{x2}{c}
\ncline{-}{x3}{e}
\ncline{-}{x4}{d}
\\[-1em]

\begin{center}
c) Well-formed, but not popular.
\end{center}
&
\begin{center}
d) Well-formed and popular.
\end{center}

\end{tabular}

\caption{An instance showing that not every well-formed matching is popular. \label{figure:instance}}

\end{center}

\end{figure}

Our alternative characterization for popular matching is based on the notion of \emph{well-formed matchings}.
\begin{definition}
A matching is well-formed if it has the following two properties: every $\first_i$-job $p$ is matched to $x \in C_i$ where $\first(x) = p$, and every applicant $x$ is matched either to $\first(x)$ or $\second(x)$. 
\end{definition}

For the unweighted case ($k=1$) our definition of well-formed matching coincides with the characterization developed by Abraham \etal \cite{AIKM05}. For $k=1$ they showed that a matching is popular if and only if is well-formed. Unfortunately, when $k>1$ not every well-formed matching is popular. For example, consider the instance in Figure~\ref{figure:instance}. There are only two well-formed matchings:  $M_1 = \{ (x_1,A), (x_2,C), (x_3,D), (x_4, E)\}$ and  $M_2 = \{ (x_1,A), (x_2,C), (x_3,E), (x_4, D)\}$. The matching $M_1$ is not popular because $\{ (x_2,A), (x_3,C), (x_4,D) \}$ is more popular than $M_1$. On the other hand $M_2$ is popular.

Nevertheless, we can still prove that being well-formed is a necessary condition for begin popular.

\begin{theorem} \label{theorem:characterization}
Let $M$ be a popular matching, then $M$ is well-formed.
\end{theorem}

One could be tempted to discard the current definition of well-formed matching and seek a stronger one that will let us replace the \emph{if then} of Theorem~\ref{theorem:characterization} with an \emph{if and only if}. As we will see, in proving the theorem we only use the fact that $w_i > w_{i+1}$. Armed with this sole fact Theorem~\ref{theorem:characterization} is the best we can hope for because if the weights are sufficiently spread apart, one can show that every well-formed matching is in fact popular; see Observation~\ref{obs:all-popular}.

The proof of Theorem~\ref{theorem:characterization} is broken down into Lemmas~\ref{lemma:first} and~\ref{lemma:second}.

\begin{lemma} \label{lemma:first} Let $M$ be a popular matching, then every $\first_i$-job $p$ is matched to an applicant $x \in C_i$ such that $\first(x) = p$ 
\end{lemma}

\begin{floatingfigure}[r]{3.75cm}
\vspace{1ex}\small
\(\begin{array}{l@{\hspace{3.5em}}l}
 \cnode*{1.5pt}{z} \rput[r](-0.25,0){z}  & \cnode*{1.5pt}{fy} \rput[l](0.25,0){\first(y)} \\[0.35cm]
\cnode*{1.5pt}{y} \rput[r](-0.25,0){y}  & \cnode*{1.5pt}{p} \rput[l](0.25,0){\first(x) = p } \\[0.35cm]
\cnode*{1.5pt}{x} \rput[r](-0.25,0){x}  \\[7em]
\end{array}\)
\ncline{-}{z}{fy}
\ncline[linestyle=dashed,dash=4pt 2pt,arrowinset=0.2]{->}{y}{fy}
\ncline{-}{y}{p}
\ncline[linestyle=dashed,dash=4pt 2pt,arrowinset=0.2]{->}{x}{p}
\end{floatingfigure}

\newcommand\myproof{ \noindent {\it Proof.}~}
\myproof By induction on $i$. For the base case let $x \in C_1$ and $\first(x)=p$. For the sake of contradiction assume that $p$ is matched to $y$ and $\first(y) \neq p$. If $y \in C_{s>1}$, then promote $x$ to $p$ and demote $y$ to $l(y)$. The swap improves the satisfaction by $w_1 - w_s > 0$, but this cannot be since $M$ is popular. If $y \in C_1$ then promote $x$ to $p$ and $y$ to $\first(y)$, and demote applicant $z = M(\first(y))$ as depicted on the right. Thus, the satisfaction improves by $w_1 + w_1 - w(z) > 0$. 

For the inductive case let $x \in C_i$ and $\first(x) = p$. Assume like before that $M(p)=y$ and $\first(y) \neq p$. If $y \in C_{s>i}$, then promote $x$ and demote $y$ to get a change in satisfaction of $w_i - w_s > 0$. If $y \in C_{s<i}$ then by induction $\first(y)$ is matched to $z \in C_s$, promoting $x$ to $p$ and  $y$ to $\first(y)$ while demoting $z$ changes the satisfaction by $w_i + w_s - w_s > 0$. Finally, suppose $y \in C_i$. Let $z = M(\first(y))$. If $z \in C_{\geq  i}$ then the usual promotions change the satisfaction by at least $w_i$. Note that if $z \in C_ {s < i}$ then $\first(z) \neq \first(y)$ by definition of $\first(y)$. Letting $y$ play the role of $x$ in the above argument handles this case.
In every case we have reached the contradiction that $M$ is not popular, therefore the lemma follows.
\qed \\

\begin{lemma} \label{lemma:second} Let $M$ be a popular matching, then every $x \in A$ is matched either to $\first(x)$ or $\second(x)$.
\end{lemma}

\begin{floatingfigure}[r]{3.75cm}
\vspace{1ex}\(\small \begin{array}{l@{\hspace{3.5em}}l}
 \cnode*{1.5pt}{z} \rput[r](-0.25,0){z}  & \cnode*{1.5pt}{fy} \rput[l](0.25,0){\first(y)} \\[0.35cm]
\cnode*{1.5pt}{y} \rput[r](-0.25,0){y}  & \cnode*{1.5pt}{p} \rput[l](0.25,0){\second(x) = p } \\[0.35cm]
\cnode*{1.5pt}{x} \rput[r](-0.25,0){x} \\[7em]
\end{array}\)
\ncline{-}{z}{fy}
\ncline[linestyle=dashed,dash=4pt 2pt,arrowinset=0.2]{->}{y}{fy}
\ncline{-}{y}{p}
\ncline[linestyle=dashed,dash=4pt 2pt,arrowinset=0.2]{->}{x}{p}
\end{floatingfigure}

\myproof
As a corollary of Lemma~\ref{lemma:first} no $x \in A$ can be matched to a job which is strictly better than $\first(x)$ or in between $\first(x)$ and $\second(x)$. Hence we just need to show that $x$ cannot be matched to a job which is strictly worse than $\second(x)$. For the sake of contradiction let us assume this is the case.

Let $x \in C_i$ and $p=\second(x)$. Note that $p$ must be matched to some applicant $y$, otherwise we get an immediate improvement by promoting $x$ to $p$. If $y \in C_{s>i}$ then promoting $x$ and demoting $y$ gives us a more popular matching because $w_i - w_s > 0$. Otherwise $y$ belongs to $C_{s\leq i}$ in which case $\first(y) \neq p$. By Lemma~\ref{lemma:first} there exists $z \in C_s$ matched to $\first(y)$. Promoting $x$ to $\second(x)$ and $y$ to $\first(y)$ while demoting $z$ improves the satisfaction by $w_i + w_s - w_s > 0$. A contradiction.
\qed \\

Let $G$ be a subgraph of $H$ having only those edges between applicants and their first and second jobs. See the graph in Figure~\ref{figure:instance}.b. Theorem~\ref{theorem:characterization} tells us that every popular matching must be contained in $G$. Ideally we would like every well-formed matching in $G$ to be popular; unfortunately, this is not always the case. To remedy this situation, we will prune some edges from $G$ that cannot be part of any popular matching. Then we will argue that every well-formed matching in the pruned graph is popular. In order to understand the intuition behind the pruning algorithm we need the notion of $\emph{promotion path}$.

\begin{definition} A promotion path \wrt a well-formed matching $M$ is a sequence $p_0, x_0, \ldots, p_s, x_s$, such that $p_i = \first( x_i)$,  $(x_i,p_i) \in M$, and for all $i< s$, applicant $x_i$ prefers $p_{i+1}$ over $p_i$.
\end{definition}

Such a path can be used to free $p_0$ by promoting $x_i$ to $p_{i+1}$, for all $i < s$, and demoting $x_s$. We say the cost (in terms of satisfaction) of the path is $w(x_s) - w(x_0) - \ldots - w(x_{s-1})$, as everyone gets a better job except $x_s$. To illustrate this consider the instance in Figure~\ref{figure:instance}, and the well-formed matching $\{(x_1,A), (x_2,C), (x_3,D), (x_4,E) \}$. The sequence $D, x_3, C, x_2, A, x_1$ is a promotion path with cost $w(x_1) - w(x_2) - w(x_3)=1$ that can be used to free $D$.

To see how promotion paths come into play, let $M$ be a well-formed matching and $M'$ be any other matching. Suppose $y$ prefers $M'$ over $M$, we will construct a promotion path starting at $p_0 = M'(y)$. Note that $p_0$ is an $\first$-job and must be matched in $M$ to $x_0$ such that $\first(x_0) = p_0$. Thus, our path starts with $p_0, x_0$. To extend the path from $x_i$, check if  $x_i$ prefers $M'$ over $M$, if that is the case, $p_{i+1} = M'(x_i)$ and $x_{i+1} = M(p_{i+1})$, otherwise the path ends at $x_i$. Notice that if $x_i \in C_s$ then $x_{i+1} \in C_{<s}$, which means the path must be simple. Coming back to $y$, the applicant who induced the path, note that if $w(y)$ is greater than the cost of the path, then $M$ cannot be popular because using the promotion path and promoting $y$ to $p_0$ gives us a more popular matching. On the other hand, it is easy to see that if for every applicant $y$, the cost of the path induced by $y$ is at least $w(y)$, then $M'$ cannot be more popular than $M$. For example, if the weights are sufficiently spread apart then any promotion path out of an $f_i$-job will have cost at least $w_i$, and as a result, any well-formed matching would be popular.

\begin{observation} \label{obs:all-popular} If $w_i \geq 2 w_{i+1}$ for all $i < k$ then every well-formed matching is popular.
\end{observation}

The pruning procedure keeps a label $\lab(p)$ for every $f_i$-job $p$. Based on these labels we will decide which edges to prune. The following invariant states the meaning these labels carry.

\begin{invariant} \label{invariant:label} Let $p$ be an $\first_i$-job and $M$ be any well-formed matching contained in the pruned graph. A minimum cost promotion path out of $p$ \wrt $M$ has cost exactly~$\lab(p)$.
\end{invariant}

We now describe the procedure {\sc pruned-strict} whose pseudo-code is given in Figure~\ref{figure:pruning-strict}. The algorithm works in iterations. The $i$th iteration consists of two steps. First, we prune some edges incident to $C_i$, making sure that these edges do not belong to any popular matching. Second, we label all the $f_i$-jobs so that Invariant~\ref{invariant:label} holds for them. Note that later pruning cannot break the invariant for $f_i$-jobs as promotion paths out of these jobs only use edges incident to applicants in $C_{\leq i}$.

In the first iteration we do not prune any edges. Notice that a promotion path out of an $f_1$-job must end in its $C_1$ mate, therefore line 1 sets the label of all $\first_1$-jobs to $w_1$.

At the beginning of the $i$th iteration we know that the invariant holds for all \mbox{$\first_{< i}$-jobs}. Consider an applicant $x \in C_i$. Let $q$ be a job $x$ prefers over $\first(x)$. Note that $q$ must be an $\first_{<i}$-job, therefore, in any well-formed matching included in the pruned graph the minimum cost promotion out $q$ has cost $\lab(q)$. We can use the path to free $q$ and then promote $x$ to it, the total change in satisfaction is $w_i - \lab(q)$. Therefore, if $\lab(q) < w_i$ then no popular matching exists. Lines 3--5 check for this. The expression $\minlabel(x, r)$ is a shorthand notation for $\min_q \lambda(q)$ where $q$ is a job $x$ prefers over $r$; if there is no such job, we define $\min_q \lambda(q) = \infty$.

\begin{figure}[t]
\center
\psframebox[boxsep=true,framearc=0.05,framesep=8pt]{
\begin{minipage}{2cm}
\begin{algorithm}{prune-strict}{G}
\mbox{All $f_1$-jobs get a label of $w_1$.} \\
\begin{FOR}{i = 2 \TO k}
   \begin{FOR}{x \in C_i}
      \begin{IF}{ \minlabel(x,\first(x)) < w_i }
          \RETURN \mbox{``no popular matching exists''}
      \end{IF}
   \end{FOR} \\
   \begin{FOR}{p \in \mbox{$f_i$-job}}
      S \= \{ x \in \A\, |\, \first(x) = p \} \\
      \begin{IF}{S = \{x\} }
          \lab(p) \= \min (w_i, \minlabel(x, \first(x) ) -w_i)
      \ELSE
          \lab(p) \= w_i \\
          \begin{FOR}{x \in S \mbox{ such that } \minlabel(x,p) < 2 w_i}
                \mbox{prune the edge } (x,p)
          \end{FOR}
      \end{IF}
   \end{FOR}
\end{FOR} \\
\begin{FOR}{x \in \A \mbox{ such that } \minlabel(x,\second(x)) < w(x)}
   \mbox{prune the edge } (x, \second(x))
\end{FOR} \end{algorithm}
\end{minipage} }
\caption{Pruning the graph with strict preference lists. \label{figure:pruning-strict}}
\end{figure}

Let $p$ be an $\first_i$-job and $S$ be the set of applicants in $C_i$ whose first job is $p$, also let $M$ be a well-formed matching contained in the pruned graph. Suppose $S$ consists of just one applicant $x$, then $(x,p)$ must belong to $M$. A promotion path out of $p$ either ends at $x$ or continues with another job which $x$ prefers over $p$. Therefore $\lab(p) =\min (w_i, \minlabel(x, \first(x) )-w_i) $, which must be non-negative. On the other hand, if $|S|>1$ then only one of these applicants will be matched to $p$ while the rest must get their second job. Suppose $M(p) = x \in S$. Invariant~\ref{invariant:label} tells us that there exists a promotion path \wrt $M$ out of $p$ with cost $\minlabel(x,p) - w_i$ that can be used to free $p$, which in turn allows us promote one of the other applicants in $S - x$ to $p$. Therefore if $\minlabel(x,p) < 2w_i$, $M$ is not popular, which means the edge $(x, p)$ cannot belong to any popular matching and can safely be pruned. We set $\lab(p) = w_i$ because in the pruned graph $p$ can only be matched to $x \in S$ such that $\minlabel(x,p) \geq 2w_i$. Lines 6--12 capture exactly this.

Finally, Lines 13--14 prune edges $(x,\second(x))$ that cannot be part of any popular matching because of promotion paths out of jobs between $\first(x)$ and $\second(x)$ on $x$'s list, with cost $\minlabel(x,\second(x)) < w_i$.

To exemplify how {\sc prune-strict} works let us run the algorithm on the instance in Figure~\ref{figure:instance}. The jobs are labeled $\lab(A) = 7$, $\lab(C) = 3$, and $\lab(D) =2$. The only edge pruned is $(x_3, D)$ because both $x_3$ and $x_4$ have $D$ as their first job and $\minlabel(x_3, D) = 3 < 4 = 2 w(x_3)$. Hence, the only well-formed matching included in the pruned graph is $M_2 = \{ (x_1,A), (x_2,B), (x_3,E),(x_4,D)\}$. The next theorem states that $M_2$ must be popular.

\begin{theorem} Let $G'$ be the resulting pruned graph after running {\sc prune-strict}. Then a matching is popular if and only if it is a well-formed matching in $G'$.
\end{theorem}

We have argued that no pruned edge can be present in any popular matching, let us now show that every well-formed matching $M$ in the pruned graph is indeed popular. Let $M'$ be any other matching, our goal is to show that $M'$ is not more popular than $M$. Suppose $x$ prefers $M'$ over $M$, this induces a promotion path at $M'(x)$ with respect to $M$. If $x$ gets $f(x)$ in $M$ then the cost of such a path is at least $\minlabel(x, \first(x)) \geq w_i$ by Lines 4-5. Otherwise, $M(x)=\second(x)$ and Lines 13-14 make sure the cost at the promotion path is at least $w_i$. Since this holds for every applicant $x$, $M'$ cannot be more popular than $M$. 

It is entirely possible that the pruned graph does not contain any well-formed matching. In this case we know that no popular matching exists.

\subsection{Implementation}

Let $G$ be the graph with edge set $\{ (x, \first(x))$, $(x, \second(x)) \, |\, x \in \A \}$. Assuming the applicants are already partitioned into categories $C_i$, we can compute $G$ in $O(n+m)$ time. The pruning procedure also takes linear time, since the $i$th iteration takes $O\left(\sum_{x \in C_i} \deg_H(x)\right)$ time. Let $G'$ be the pruned graph. Finding a popular matching reduces to finding a well-formed matching in $G'$.

Abraham \etal \cite{AIKM05} showed how to build a well-formed matching for unweighted instances ($k=1$), if one exists, in linear time. The unweighted setting is slightly simpler than ours. More specifically, the set of second jobs is disjoint from the set first jobs and every applicant in $G'$ has degree exactly 2. These two issues can be easily handled: First, for every edge $(x,\second(x))$ in $G$ if $s(x)$ happens to be someone else's first job then prune the edge $(x,s(x))$. Second, iteratively find an applicant $x$ with degree 1, let $p$ be $x$'s unique neighbor, add $(x,p)$ to the matching, and remove $x$ and $p$ from $G'$. All these modifications can be carried out in $O(n+m)$ time. If at the end some applicant has degree 0 there is no well-formed matching, and consequently no popular matching. Otherwise every applicant has degree 2 and the set of $f$-jobs is disjoint from the set of $s$-jobs, thus we can apply directly the linear time algorithm of Abraham \etal \cite{AIKM05}.

\begin{theorem} In the case of strict preferences lists, we can find a weighted popular matching, or determine that none exists, in $O(n + m)$ time. \label{theorem:strict}
\end{theorem}

Recall that at the beginning we modified the instance by adding dummy jobs at the end of every applicant's list. A natural objective would be to find a popular matching that minimizes the number of applicants getting a dummy job. The cited work also shows how to do this in $O(n+m)$ time; thus, it carries over to our problem.

\section{Preference lists with ties}

Needless to say, if ties are allowed in the preference lists, the solution from the previous section does not work anymore. We will work out an alternative definition for first and second jobs which will lead to a new definition of well-formed matchings. Like in the case without ties if a matching is popular then it must be well-formed, but the converse does not always hold. We will show how to prune some edges that cannot be part of any popular matching to arrive at the goal that every well-formed matching in this pruned graph is popular.

Let us start by revising the notion of first job. For $x \in C_1$, let $\first(x)$ be the set of jobs on $x$'s list with the highest rank. Let $G_1$ be the graph with edges between applicants in $C_1$ and their first jobs. We say a job/applicant is \emph{critical} in $G_i$ if it is matched in every maximum matching of $G_i$, otherwise we say it is \emph{non-critical}. For $x \in C_i$, define $\first(x)$ as the highest ranked jobs on $x$'s list which are non-critical in all $G_{< i}$. The graph $G_i$ includes $G_{i-1}$ and edges between applicants in $C_i$ and their first jobs. We note that a critical node in $G_i$ may be non-critical in some $G_{>i}$.

If $x \in C_i$ is non-critical in $G_i$ we define $\second(x)$ as the highest ranked set of jobs on $x$'s list which are non-critical in all $G_{\leq i}$. If $x$ is critical in $G_i$ then $\second(x)$ is the empty set.

\begin{observation} For every applicant $x \in A$ we have $\first(x) \cap \second(x) = \emptyset$.
\end{observation}

Essentially, when $x \in C_i$ is non-critical in $G_i$ we can show that all the jobs in $\first(x)$ are critical, therefore $\first(x)$ and $\second(x)$ are always disjoint.

\begin{definition}
A matching $M$ is \emph{well-formed} if, for all $1 \leq i \leq k$, the matching $M_i = M \cap E[G_i]$ is maximum in $G_i$, and every applicant $x$ is matched within $\first(x) \cup \second(x)$. 
\end{definition}

Notice that when there are no ties all these definitions are identical to the ones given in the previous section. Before proceeding to prove Theorem~\ref{theorem:characterization} in the context of ties we review some basic notions of matching theory.

The following definitions are all with respect to a given matching $M$. An \emph{alternating path} is a simple path that alternates between matched and free edges. An \emph{augmenting path} is an alternating path that starts and ends with a free vertex. An \emph{exchange path} is an alternating path that start with a matched edge and ends with a free vertex. We can update $M$ along an augmenting or exchange path $P$ to get the matching $M \oplus P$, the symmetric difference of $M$ and $P$. 

In our proofs we will make use of the following property of non-critical nodes, which is a part of the Gallai-Edmonds decomposition \cite{S03}.

\begin{lemma} \label{lemma:critical} Let $G$ be a bipartite graph and let $v$ be a non-critical vertex. Then, in every maximum matching $M$ of $G$ there exists an alternating path starting at $v$ and ending with a free vertex.
\end{lemma}

\begin{proof}
If $v$ is free in $M$, the lemma is trivially true, so assume that $v$ is matched in $M$.
Since $v$ is non-critical there is a maximum matching $O$ in which $v$ if free. In $O \oplus M$ there must be an alternating path \wrt $M$ of even length that starts at $v$ and ends with a vertex free in $M$.
\end{proof}

The next two lemmas prove Theorem~\ref{theorem:characterization} under the new definition of well-formed.

\begin{lemma} \label{lemma:first-ties} Let $M$ be a popular matching. Then, for all $i$, $M_i = M \cap E[G_i]$ is maximum in $G_i$.
\end{lemma}

\begin{proof} 
By induction on $i$. For the base case, suppose that $M_1$ is not maximum, then there must be an augmenting path in $G_1$ \wrt $M_1$ starting at $x \in C_1$ and ending at $p$. If $p$ is free in $M$ or $p = M(x)$ then we can update $M$ along the path\footnote{While $P$ is augmenting \wrt $M_i$, it may not be augmenting \wrt $M$ since $x$ could be matched in $M$. This can be easily fixed by removing $(x,M(x))$ from $M$ before doing the update. For the sake of succinctness, from now on we assume that such implicit fix always occurs when updating along a path that ends with a free edge leading to a matched node.} to improve the satisfaction by $w_1$, so let us assume that there exists $y = M(p) \neq x$. If $y \in C_{s>1}$ then updating $M$ along $P$ gives us an improvement in satisfaction of $w_1 - w_s > 0$. Suppose then that $y \in C_1$, and let $q$ be a job in $\first(y)$. Since $(y,p) \notin M_1$, applicant $y$ must prefer $q$ over $p$. If the $q$ belongs to $P$ then we can create an alternating cycle by replacing the section of $P$ before $q$ with the edge $(y,q)$. Updating the matching along the cycle improves the satisfaction by $w(y)$. If $q$ does not belong to $P$ then appending $(y,q)$ to $P$ and updating along the resulting path changes the satisfaction by $w_1+w_1 - w(M(q)) > 0$. In every case we reach the contraction that $M$ is not popular, thus $M_1$ must be maximum in $G_1$.

For the inductive step, if $M_i$ is not maximum we can find like before an augmenting path $P$ starting at $x \in C_i$ and ending at a job $p$. If $p$ is free in $M$ or $M(p) \in C_{>i}$ updating along $P$ improves the satisfaction, so assume that $p$ is matched in $M$ to $y \in C_{s \leq i} -x$. Let $q$ be a job in $\first(y)$, if $q$ belongs to $P$ then we can construct an alternating cycle to improve the satisfaction, so assume that $q \notin P$. Since $(y,p) \notin M_i$, we know that $p \notin f(y)$; by inductive hypothesis $p$ must be strictly worse than $q$. We update $M$ along $P$ to get $M'$; note that $y$ is free in $M'$. By inductive hypothesis $M_s$ is maximum in all $G_{s<i}$, therefore $M'_s$ is maximum in $G_{s}$ as well. There are three cases to consider. First, if $y \in C_1$ then we can promote $y$ to $q$ and demote whoever is matched to $q$, the total change in satisfaction is $w_i + w_1 - w(M'(q)) > 0$. Second, consider the case $y \in C_{1<s<i}$. Note that $q$ cannot be free in $M'_s$, as this would contradict the maximality of $M'_s$. If $M'(q) \in C_s$ we are done since promoting $y$ to $q$ gives a total change in satisfaction of $w_i + w_s - w_s > 0$  \wrt $M$, so assume $M'(q) \in C_{<s}$. By definition of $\first(y)$, $q$ is non-critical in $G_{s-1}$. Thus we can find an exchange path~$Q$ \wrt $M'_{s-1}$ starting at $q$ and ending at a job $r$ free in $M'_{s-1}$. Note that $r$ cannot be free in $M'_s$, otherwise $M'_s$ would not be maximum in $G_s$, thus $M(r) \in C_s$. Updating $M'$ along $Q$ to free $p$ and promoting $y$ to $p$ gives us a new matching $M''$. The satisfaction of $M''$ \wrt $M$ is $w_i + w_s - w_s > 0$, thus $M$ cannot be popular.

Finally, we need to consider the most involved case, namely, $y \in C_i$. Note that we cannot use the argument given above because $M'_i$ need not be maximum in $G_i$. In order to fix this let us forget about $M'$ and consider a matching $O$ maximum in $G_i$; furthermore, assume $O$ minimizes $|M_i \oplus O|$. The set $M_i \oplus O$ is made up of paths $P_1, P_2, \ldots, P_a$, each of which is augmenting \wrt $M_i$. By inductive hypothesis $M$ is maximum in $G_{<i}$ so each $P_j$ starts at $x_j \in C_i$ and end at some $p_j$, both free in $M_i$. Let $y_j$ be $M(p_j)$ and $q_j$ be a job in $f(y_j)$. Assume $y_j \in C_i - \{x_1, \ldots,x_a\}$ and $q_j \notin P_j$, otherwise we fall in one of the cases we have already covered. Now suppose that $q_1$ belongs to some path $P_h$ for $h\neq 1$; then we can replace the portion of the path of $P_h$ before $q_1$ with the edge $(y_1,q_1)$ update $M$ along the resulting path and then update along $P_1$ to improve the satisfaction by $w_i + w_i - w_i> 0$. Thus we can assume that $q_1 \notin P_h$ for all $h$. At this point we can safely update $M$ along all paths $P_1, \ldots, P_a$ to get a matching $M'$, which is maximum in $G_{\leq i}$. Finally, we can use the argument above on $y_1$. Namely, find an exchange path $Q$ \wrt $M_{i-1}$ from $q_1$ to $r$; if $r \neq M'(x_i)$ for all $i$ then the same argument applies. On the other hand, if $r = M'(x_i)$ for some $i$ then give $x_i$ its original job $M(x_i)$, update the matching along $Q$ and promote $y_1$ to $q$, which improves the satisfaction by $w_i$. Notice that in the last exchange we assumed $M(x_i)$ was free in $M'$, or equivalently, that $M(x_i) \neq p_j$ for all $j$ and $i$. Indeed, if $M(x_i) = p_j$ for some $j$ and $i$ then we can join together $P_j$ and $P_i$ using the edge $(x_i, p_j)$ and update $M$ along the resulting path to improve the satisfaction by $w_i + w_i + w_i - w_i > 0$.

In every case we have reached a contradiction, thus the lemma follows.
\end{proof}

\begin{lemma} \label{lemma:second-ties} Let $M$ be a popular matching, then every applicant $x$ is matched within $\first(x) \cup \second(x)$.
\end{lemma}

\begin{proof} Recall that $\second(x)$ is undefined only if every job in $x$'s preference list is critical in some $G_{\leq i}$. But $l(x)$, $x$'s the last resort job, is critical in $G_i$ if and only if $\first(x) = \{ l(x) \}$, in which case $M(x) = l(x)$ by Lemma~\ref{lemma:first-ties}. Let us assume then that $s(x)$ is well defined.

For the sake of contradiction assume that the lemma does not hold for $x \in C_i$. Note that all jobs which $x$ prefers over $\second(x)$ are critical in $G_i$, among these, only $\first(x)$ have an edge to $x$ in $G_i$. Thus if Lemma~\ref{lemma:first-ties} is to hold, $M(x)$ must be strictly worse than any job $p \in \second(x)$.

Consider the applicant $y = M(p)$. If $y \in C_{s> i}$ then $y$ can be demoted and $x$ promoted to $p$ to improve the satisfaction by $w_i - w_s > 0$. If $y \in C_{1<s\leq i}$, the job $p$ is strictly worse than any job $q \in \first(y)$. Using the fact that $q$ is non-critical in $G_{s-1}$ we find an alternating path in $M_{s-1}$ to a free (in $M_{s-1}$) job $r$ which must be matched in $M_s$ to $z \in C_s$. Updating along the path and promoting $x$ to $p$ improves the satisfaction by $w_i + w_s - w_s > 0$. The case where $y \in C_1$ is simpler as we can promote $y$ to $q$ and demote whoever was matched to $q$; the change in satisfaction is $w_i + w_1 - w(M(q)) > 0$.
\end{proof}

This finishes the proof of Theorem~\ref{theorem:characterization} under the new definition of well-formed matching. Thus every popular matching is contained in $G$, the graph consisting of those edges between applicants and their first and second jobs. Because the new definition of well-formed matching generalizes the one for strict preferences, we again encounter the problem that not every well-formed matching is popular. We proceed as before, pruning certain edges which are not part of any popular matching. Finally, we show that every well-formed matching in the pruned graph is popular.

It is time to revise the definition of promotion path. Let $M$ be a well-formed matching. Our promotion path starts at $p_0$, a job critical in $G_{i_0}$, but non-critical in all $G_{< i_0}$. We find an alternating path in $G_{i_0}$ \wrt $M_{i_0}$ from $p_0$ to $x_0$ which starts and ends with a matched edge; we augment along the path to get $M'$. Let $p_1$ be a job which according to $x_0$ is better than $\first(x_0)$ (or as good, but not in $\first(x_0)$), moreover let $p_1$ be critical in $G_{i_1}$, but non-critical in all $G_{<i_1}$. Since $x_0 \in C_{>i_1}$, the matching $M'_{j}$ is still maximum in $G_{\leq i_1}$. Find a similar alternating path in $G_{i_1}$ \wrt $M'_{i_1}$ from $p_1$ to $x_1$, update $M'$, and so on. Finally, every applicant $x_i$ is assigned to $p_{i+1}$, except $x_s$, the last applicant in the path, who is demoted. The cost of the path is defined as the satisfaction of $M$ with respect to $M'$, or equivalently, $w(x_s)$ minus the weight of those applicants $x_{i<s}$ who find $p_{i+1}$ strictly better than $\first(x_i)$ (recall that $p_{i+1}$ may be as good as, but not in $\first(x_i)$). This is the price to pay, in terms of satisfaction, to free $p_0$ using the path.

To see why this is the right definition, let $M$ be a well-formed matching and $M'$ be any other matching. Suppose $y$ prefers $M'$ over $M$, we will construct a promotion path starting at $p_0 = M'(y)$. Since $M$ is well-formed, $p_0$ must be critical; let $i_0$ be the smallest $i$ such that $p_0$ is critical in $G_i$. Taking $M_{i_0} \oplus M'_{i_0}$ we can find an alternating path that starts with $(p_0, M(p_0) )$ and ends at $x_0$ which is free in $M'_{i_0}$---the path cannot end in a job that is free in $M_{i_0}$ because $p_0$ is critical. Either $x_0$ gets a worse job under $M'$, in which case the promotion path ends, or gets a job $p_1$ which is better than $\first(x_0)$, or just as good but does not belong to $\first(x_0)$. We continue growing the path until we run into an applicant $x_s$ who prefers $M$ over $M'$, notice that since $i_j > i_{j+1}$ we  are bound to find such an applicant. Now, if the cost of the path is less than $w(y)$ then we know the well-formed matching $M$ is not popular. On the other hand, if the cost of the path induced by $y$ is at least $w(y)$, for all such $y$, we can claim that $M'$ is not more popular than $M$.

We are ready to discuss the algorithm {\sc prune-ties} for pruning the graph in the presence of ties, which is given in Figure~\ref{figure:pruning-ties}. In the $i$th iteration we prune some edges incident to applicants in $C_i$ making sure these edges do not belong to any popular matching, and label those jobs that became critical in $G_i$ such that Invariant~\ref{invariant:label-ties} holds.

\begin{invariant} \label{invariant:label-ties} Let $p$ be a critical job in $G_i$, and let $M$ be any matching in the pruned graph and maximum in all $G_{j\leq i}$, i.e., $M_j= M \cap E[G_j]$ is maximum in $G_j$ for all $j \leq i$. A minimum cost promotion path out of $p$ \wrt $M$ has cost exactly $\lab(p)$.
\end{invariant}

In the first iteration we do not prune any edges from $G_1$. Let $p$ be a critical job in $G_1$, and $M$ be a maximum matching in $G_1$. Every alternating path \wrt $M$ out of $p$ must end in some applicant in $C_1$, therefore, Line 1 sets $\lambda(p) = w_1$.

Recall that $\minlabel(x,p) = \min_q \lab(q)$ where $q$ ranges over jobs \emph{strictly better} than $p$ in $x$'s preference list. In addition, let us define $\equivlabel(x) = \min_q \lab(x)$ where $q$ ranges over jobs not in $\first(x)$ that have the same rank as other jobs in $\first(x)$.

For the $i$th iteration we assume Invariant~\ref{invariant:label-ties} holds for those jobs critical in some $G_{<i}$. Suppose there exists an applicant $x \in C_i$ such that $\minlabel(x, \first(x)) < w_i$. Then in every well-formed matching in the pruned graph we can find a promotion path to free a job $p$ that $x$ prefers over $\first(x)$, and then promote $x$ to $p$. This improves the satisfaction by $w_i - \minlabel(x, \first(x)) > 0$. Therefore, no popular matching exists. Lines 3--5 check for this.

Consider a vertex $x \in C_{j \leq i}$ non-critical in $G_i$. We claim that if $\minlabel(x, \first(x) )  < w_j + w_i$ or $\equivlabel(x) < w_i$ then the edges $(x,\first(x))$ cannot be part of any popular matching and can thus be pruned. Indeed, let $O$ be a matching maximum in all $G_{\leq i}$ and included in the pruned graph such that $O(x) \in \first(x)$; we will show that $O$ cannot be a subset of any popular matching. Because $x$ is non-critical in $G_i$ we know there is an exchange path \wrt $O$ from $x$ to some applicant $y \in C_{s}$ such that $j \leq s \leq i$, for otherwise $O_s$ would not be maximum in $G_s$. Augment along the path to get $O'$. While the matching $O'$ may not be maximum in $G_j$ (in case $j<s$), it is still maximum in all $G_{<j}$. Invariant~\ref{invariant:label-ties} tells us we can find a promotion path to free a job $p$ that $x$ can be promoted to; note that because $y \in C_{s\geq j}$ the changes needed to free $p$ do not affect $y$. This improve the satisfaction of the matching, therefore $O$ cannot be included in any popular matching. Since the edges $(x,f(x))$ cannot be part of any popular matching they can safely be pruned. Lines 6--8 check this.

\begin{figure}
\center
\psframebox[boxsep=true,framearc=0.05,framesep=8pt]{
\begin{minipage}{2cm}
\begin{algorithm}{prune-ties}{G}
\mbox{All critical jobs in $G_1$ get a label of $w_1$.} \\
\begin{FOR}{i = 2 \TO k}
   \begin{FOR}{x \in C_i}
      \begin{IF}{ \minlabel(x, \first(x) )  < w_i }
          \RETURN \mbox{``no popular matching exists''}
      \end{IF}
   \end{FOR} \\
   \begin{FOR}{x \in C_{j\leq i} \mbox{ non-critical in } G_i} 
     \begin{IF}{ \minlabel(x, \first(x) )  < w_j + w_i \mbox{ or } \equivlabel(x) < w_i }
       \mbox{prune the edges between $x$ and $\first(x)$}
      \end{IF}
   \end{FOR} \\
   \begin{FOR}{p \mbox{ critical in } G_i \mbox{, but non-critical in } G_{<i}}
      \mbox{let } S \= \{ x \, | \, \exists \mbox{ alternating path from $x$ to $p$} \} \\
      \lab(p) \= \min_{x \in S} \big\{ w_i, \minlabel(x, \first(x) ) - w(x), \equivlabel(x) \big\}
   \end{FOR}
\end{FOR} \\
\begin{FOR}{x \in \A \mbox{ such that } \minlabel(x,\second(x)) < w(x)}
   \mbox{prune the edges between $x$ and $\second(x)$}
\end{FOR}
\end{algorithm}
\end{minipage}
}
\caption{Pruning the graph with ties.\label{figure:pruning-ties}}
\end{figure}

Finally, we must compute $\lab(p)$ for jobs $p$ that are critical in $G_i$, but non-critical in all $G_{<i}$. A promotion path out of $p$ must begin with an alternating path starting and ending with a matched edge, going from $p$ to some applicant $x$. Since $p$ is non-critical in $G_{i-1}$ there must be alternating path in $G_i$ to some applicant in $C_i$, thus $\lab(p) \leq w_i$. Note that if $x \in C_j$ is non-critical then $\minlabel(x, \first(x) )  \geq w_j + w_i$ and $\equivlabel(x) \geq w_i$, otherwise the edges $(x, \first(x))$ would have been pruned earlier. We shall explore alternating paths out of $p$ into $x$ that start and end with a matched edge in some arbitrary matching $M_i$ maximum in $G_i$. In fact we only care about reaching applicants $x$ critical in $G_i$. Since $M_i$ is an arbitrary matching, we must argue that a similar path can always be found in any matching $O$ included in the pruned graph, maximum in all $G_{\leq i}$. To show this, augment along the path to get $M'_i$, the resulting matching is not maximum in $G_i$ any more. Take $M'_i \oplus O$, and consider the alternating path out of $p$. This path must end at an applicant $y$, matched in $O_i$, but free in $M'_i$---otherwise, if it ends in a job free in $O_i$ then $p$ is non-critical. For the sake of contradiction suppose that $y \neq x$. Since $x$ is critical, there must be a path in $M'_i \oplus O_i$ from $x$ to $z$, such that $z$ is free in $O_i$; which contradicts the fact that $x$ is critical in $G_i$. Thus we set $\lab(p)$ to the minimum of $w_i$, $\minlabel(x, \first(x) ) - w(x)$ and $\equivlabel(x)$, for those applicants $x$ that can be reached from $p$ with an alternating path starting and ending with a matched edge. Lines 9--11 do this.

The last thing to consider are non-critical applicants $x$ who may get their second job. We can promote them to a job $p$ strictly better than $\second(x)$ and start a promotion path from there. If such exchange improves the satisfaction then the edges $(x, \second(x))$ must be pruned. This is done in Lines 12--13.

\begin{theorem} Let $G'$ the resulting graph after running {\sc prune-ties}. Then a matching $M$ is popular if and only if $M$ is well-formed and $M \subseteq G'$.
\end{theorem}

We have shown that if there exists a popular matching it must be well-formed and be contained in $G'$. The proof that every well-formed matching $M$ in the pruned graph is popular is similar to that for strict preferences. Let $M'$ be any other matching, we argue that $M'$ is not more popular than $M$. Suppose $x$ prefers $M'$ over $M$, this induces a promotion path out of $M'(x)$ with respect to $M$. If $x$ gets $f(x)$ in $M$ then the cost of such a path is at least $w_i$. Otherwise, $M(x)=\second(x)$ and Lines 12-13 make sure the cost of such promotion path is at least $w_i$. Since this holds for every applicant $x$, $M'$ cannot be more popular than $M$. 

So far we have been concerned with showing the correctness of the algorithm, in the next section we show how to implement these ideas efficiently.

\subsection{Implementation}

First we need to compute $\first(x)$ and $\second(x)$ for every applicant $x \in A$; we do so in iterations. For $x \in C_1$ computing $\first(x)$ is trivial. Now build $G_1$ and find a maximum matching $M_1$ in $G_1$. Using the algorithm of Hopcroft and Karp \cite{HK73} this can be done in $O(\min(\sqrt{n},|M_1|)m)$ time. The set of critical jobs in $G_1$ can be computed in $O(m)$ time by growing a Hungarian tree \cite{CombOpt-book} from those jobs that are free in $M_1$: By Lemma~\ref{lemma:critical} those jobs that are reachable from a free job by an alternating path must be non-critical and those jobs that are not reachable from any free job must be critical. Using this information compute $\second(x)$ for all $x \in C_1$ and $\first(y)$ for all $y \in C_2$. Now construct $G_2$, augment $M_1$ to get a maximum matching $M_2$ in $G_2$, and so on. Using Hopcroft-Karp to compute $M_{i+1}$ from $M_i$ takes $O( \min( \sqrt{n} , |M_{i+1}| - |M_{i}| ) m )$ time. Adding up over all categories we get overall $O( \min( k \sqrt{n}, n) m)$ time.

The next lemma argues that {\sc prune-ties} can be implemented to run in $O(k m + n \log n)$ time. The procedure makes use of the matchings $M_1, \ldots, M_k$ found while computing $f(x)$ and $s(x)$ and the list of critical jobs in each $G_i$.

\begin{lemma} Given a matching $M_i$ maximum in $G_i$, the $i$th iteration of {\sc prune-ties} can be carried out in $O(m + |C_i| \log n)$ time.
\end{lemma}

\begin{proof}
At the beginning of the $i$th iteration we have available $\lab(p)$ for all jobs that are critical in some $G_{< i}$. Using this information it is easy to compute $\minlabel(x,\first(x))$ and $\equivlabel(x)$ in $O(\deg_H(x))$ time for each $x \in C_i$. With this information, Lines 3--8 can be done in $O(m)$ time.

Note that for each $p$ critical in $G_i$, but non-critical in all $G_{<i}$, Lines 10-11 can be implemented in linear time: Grow a Hungarian tree \wrt $M_i$  out of $p$, keep track of the applicants $x$ that can be reached from $p$, and find the one minimizing $\min \{w_i, \minlabel(x,\first(x)), \equivlabel(x) \}$. But we would like to carry out this computation \emph{for all} such jobs within the same time bounds. This can be done provided the applicants $x$ are sorted in non-decreasing value of $\min \{w_i, \minlabel(x,\first(x)), \equivlabel(x) \}$: Instead of growing Hungarian trees from the jobs we grow Hungarian trees from the matched applicants in sorted order. When growing a tree out of applicant $x$ we mark the nodes we visit and do not explore edges that lead to nodes that have already been marked. Suppose that job $p$ critical in $G_i$ was marked by applicant $x$ then clearly $\lab(p) = \min \{w_i, \minlabel(x,\first(x)), \equivlabel(x) \}$. Because a node is never explored after it has been marked, the overall work is $O(m)$. If we have a sorted list of applicants in $C_{<i}$ adding the applicants $C_i$ takes $O(|C_i| \log n)$ time if we maintain the list using a balanced search tree.
\end{proof}

Finally, after $G'$ is computed and pruned we must find a well-formed matching in it. This problem can be reduced to finding a rank-maximal matching which can be done in time $O( \min(k \sqrt{n}, n) m)$ \cite{IKMMP06}. Edges between $x \in C_i$ and $\first(x)$ get a rank of $i$, and edges from applicants to their second job get a rank of $k+1$. If the resulting rank-maximal matching is well-formed, i.e., applicant complete and maximum in all $G_i$ graphs, then we have a popular matching, otherwise no popular matching exists.

\begin{theorem} In the presence of ties we can find a weighted popular matching or determine that none exists in $O( \min(k \sqrt{n}, n) m)$ time. \label{theorem:ties}
\end{theorem}

Finding a popular matching of maximum cardinality, i.e., one that minimizes the assignment of dummy last-resort jobs, can the done within the same time bounds. Note that $\first(x) = \{ l(x) \}$ then the pair $(x,l(x))$ will be in every well-formed matching so there is no point in minimizing these edges. If $\second(x) = \{ l(x) \}$ we can give the edge $(x, l(x))$ a rank of $k+2$. Finding a rank-maximal matching in the new instance gives us a popular matching with maximum cardinality.

\section{Conclusion}

We have developed efficient algorithms for finding weighted popular matchings, a natural generalization of popular matchings. It would be interesting to study other definitions of the \emph{more popular than} relation. For example, define the satisfaction of $M$ over $R$ to be the sum (or any linear combination) of the differences of the ranks of the jobs each applicant gets in $M$ and $R$. Finding a popular matching under this new definition can be reduced to maximum weight matching, and vice versa. Defining the satisfaction to be a positive linear combination of the sign of the differences we get weighted popular matchings. We leave as an open problem to study other definitions that use a function ``in between'' these two extremes. Ideally, we would like to have efficient algorithms that can handle any odd step function.

\hspace{0.5cm}

{\bf Acknowledgment}: Thanks to David Manlove and Elena Zotenko for useful comments. Special thanks to Samir Khuller for suggesting the notion of weighted popular matchings and providing comments on earlier drafts.

\bibliographystyle{abbrv}
\bibliography{references,conferences}

\end{document}